\documentclass[11pt]{article}

\usepackage{hyperref}
\usepackage{latexsym}
\usepackage{amsmath,amssymb,amsthm}

\usepackage{setspace}
\newcommand{\beqn}{\begin{eqnarray}} 
\newcommand{\eeqn}{\end{eqnarray}} 
\newcommand{\be}{\begin{equation}} 
\newcommand{\ee}{\end{equation}} 
\newcommand{\ba}{\begin{array}} 
\newcommand{\ea}{\end{array}} 
\newcommand{\R}{{\rm\bf R}} 
\newcommand{\C}{{\rm\bf C}} 
\newcommand{\pa}{\partial} 
 
\newcommand{\re}{\ref} 
\newcommand{\ci}{\cite} 
\newcommand{\la}{\label} 

\newcommand{\fr}{\frac} 
\newcommand{\ov}{\overline} 
\newcommand{\ti}{\tilde} 

\newcommand{\si}{\sigma} 
\newcommand{\al}{\alpha} 
 
\newcommand{\dv}{{\rm div\,}} 
\newcommand{\eff}{{\rm eff}}

\newcommand{\ds}{\displaystyle} 
\newcommand{\ve}{\varepsilon} 
\newcommand{\rw}{\rm w} 
\newcommand{\curl}{{\rm curl\,}}

\newcommand{\cO}{{\cal O}}
\newcommand{\cH}{{\cal H}}
\newcommand{\cK}{{\cal K}}

\newcommand{\Y}{{\mathbb{Y}}}
\newcommand{\F}{{\mathbb{F}}}

\newcommand{\we}{\wedge} 
\newcommand{\de}{\delta} 
\newcommand{\De}{\Delta} 
\newcommand{\om}{\omega}

\newcommand{\na}{\nabla} 
 
\newcommand{\Lam}{\Lambda} 
 
\newcommand{\br}{|\kern-.25em|\kern-.25em|} 
\newcommand{\brr}{{|\kern-.15em|\kern-.15em|\kern-.15em}\,} 
\newcommand{\ddd}{\st{.\kern-.07em.\kern-.07em.}} 
\def\N{{\rm I\kern-.1567em N}}                              
\def\R{{\rm I\kern-.1567em R}}                              
\def\C{{\rm C\kern-4.7pt                                    
\vrule height 7.7pt width 0.4pt depth -0.5pt \phantom {.}}} 
\def\Z  {{\sf Z\kern-4.5pt Z}}                                
 
\begin{document} 
 \renewcommand{\theequation}{\thesection.\arabic{equation}} 
\newtheorem{theorem}{Theorem}[section] 
\renewcommand{\thetheorem}{\arabic{section}.\arabic{theorem}} 
\newtheorem{definition}[theorem]{Definition} 
\newtheorem{deflem}[theorem]{Definition and Lemma} 
\newtheorem{lemma}[theorem]{Lemma} 
\newtheorem{example}[theorem]{Example} 
\newtheorem{remark}[theorem]{Remark} 
\newtheorem{remarks}[theorem]{Remarks} 
\newtheorem{cor}[theorem]{Corollary} 
\newtheorem{pro}[theorem]{Proposition}

\begin{center}
{\Large On  global attraction to solitons for   
3D Maxwell-Lorentz equations }
 \bigskip\\ 
{\large E.A. Kopylova, A.I. Komech}
\end{center}

\begin{abstract} 
We consider the Maxwell field  coupled to a single rotating charge. 
This Hamiltonian system admits soliton-type solutions, 
where the field is static, while the charge rotates with constant angular velocity. 
We prove that any solution of finite energy converges, in suitable local energy seminorms, to the corresponding
soliton in the long time limit $t\to\pm\infty$. 
\end{abstract}

 
\setcounter{equation}{0} 
 
\section{Introduction} 
We consider the Abraham--Lorentz model  describing a motion of  spinning extended charged particle in the Maxwell field $(E(x,t),B(x,t))$
(see \ci[Chapter 10]{S2004}).
We  restrict ourselves  to  the situation, where  the spinning  particle is located at the origin. 
This can be achieved by assuming the (anti-) symmetry conditions
$E(-x)=-E(x)$, $B(-x)=B(x)$  for the initial fields. Then this property persists for all times.
In this case the Maxwell--Lorentz equations  read
\begin{equation} \la{mls1}
\left\{\begin{array}{rcl}
\dot E(x,t)\!\!&\!\!=\!\!&\!\!\curl B(x,t)-[\om(t)\we x]\rho(x),\qquad \dot B(x,t)= - \curl E(x,t)
\\
\dv E(x,t)\!\!&\!\!=\!\!&\!\!\rho (x),\qquad\qquad\qquad\quad \quad \quad\,\,
\dv B(x,t)=0
\\
 I\dot \om(t)\!\!&\!\!=\!\!&\langle x\we \big[E(x,t)+(\om(t)\we x)\we B(x,t)\big],\rho(x)\rangle
\end{array}\right|, \qquad
\end{equation}
Here  $m$ is  the mass of the particle,  $\rho(x)$ is the charge distribution, $\om(t)$ is the angular velocity,
$$
I=\frac 23m_b\int|x|^2\rho(x) dx>0
$$
is the bare moment of inertia associated to the bare mass $m_b$; all other constants are set equal to unity. 
The brackets $\langle,\rangle$ denote the inner product in the real Hilbert space $L^2:=L^2(\R^3)\otimes\R^3$. 
We assume that   real-valued charge density $\rho(x)\not\equiv 0$ is smooth and    spherically-invariant, i.e.,
\be\la{rosym}
\rho\in C_0^\infty(\R^3),\quad\rho(x)=\rho_{rad}(|x|),\quad \rho(x)=0~~ {\rm for}~~|x|\ge R_\rho>0.
\ee
For any $\om\in R^3$ the system \eqref{mls1} admits a stationary state (solitons) $(E_\om(x),B_\om(x),\om)$ \cite{IKS2004, S2004}. 
We denote by ${\cal S}:=\{(E_\om(x),B_\om(x),\om),~\om\in\R^3\}$
the set of all solitons.

Let  $\mu_0=0$, $\mu_{-j}=-\mu_{j}$,  $j\in\N$, be   zeros of the function 
$$
g(\mu)=i\sqrt{\frac{2}{\pi}}\int_{0}^{\infty}\frac{\mu r\cos (\mu r)-\sin(\mu r)}{\mu^2}\rho_{rad}(r)\,rdr, \quad \mu\in\R.
$$
The  set of zeros is (at most) countable and does not contains  accumulation points (see \cite{Kunze}).
We suppose that  
\be\la{mm}
\mu_j+\mu_k\ne \mu_\ell ~~{\rm for} ~~(j,k,\ell)\not\in \{(-n,n,0), (0,n,n), (n,0,n)\}.
\ee
We assume that  for some  $R_0 > 0$ 
the initial fields $E_0(x), B_0(x)\in C^2(\R^3\setminus {\rm B}_{R_0})$, and  
 \beqn\nonumber 
&&|x|\big(|E_0(x)|+|B_0(x)|\big) +|x|^2\big(|\nabla E_0(x)|+|\nabla B_0(x)|\big) ={\cal O}(|x|^{-\sigma}), \\
\la{8}  
 &&|\na\na E_0(x)|+|\na\na B_0(x)|= {\cal O}(|x|^{-1}),\quad x\to\infty
 \eeqn
 with  some  $\si>1/2$. Here $B_{R_0}$ is the ball $\{x\in\R^3:\,|x|\le R_0\}$.

Under these conditions, Kunze  proved (see \cite[Theorem 1.4]{Kunze}))  the relaxation of angular acceleration:
\be\la{dot-om-decay} 
\dot\om(t)\to 0, \quad \ddot\om\to 0, \qquad t\to\infty
\ee
for $m_b\not\in{\cal M}_\rho$, where ${\cal M}_\rho$ is an at most countable set.   
The set $\cal M_\rho$  is defined in \cite[(6.1)]{Kunze}.
For example, for the uniformly charged ball,  ${\cal M}_\rho= \{4\pi(\mu_j^2-30)\mu_j^{-4}, j\in\Z\}$.
The convergence  \eqref{dot-om-decay} implies the convergence of  solutions to the set $\cal S$ of all solitons:
\be\la{K-res}
\inf\limits_{\ti\om\in\R^3}\big(|\om(t)-\ti\om|+\Vert E(t)-E_{\ti\om}\Vert_{L^2(B_R)}+\Vert B(t)-B_{\ti\om}\Vert_{L^2(B_R)}\big)\to 0,\quad t\to\infty,
\ee
for any $R>0$.
The main result of the present paper is the convergence to a particular soliton:
\be\la{EK-res}
|\om(t)-\om_{\pm}|+\Vert E(t)-E_{\om_{\pm}}\Vert_{L^2(B_R)}+\Vert B(t)-B_{\om_{\pm}}\Vert_{L^2(B_R)}\to 0,\quad t\to\pm\infty,
\ee
for some $\om_{\pm}\in\R^3$ depend on the solution.
To prove the convergence \eqref {EK-res} 
we  combine 
the limit  \eqref{dot-om-decay} with the orbital stability established in  \cite{KK2023}.
The orbital stability ensures that the solution with  initial data close to a  soliton remain so. 

Let us comment on our approach.
We rewrite the system \eqref{mls2} in Maxwell potentials ( equations  \eqref{mls2}).
Then for the field part $(A,\dot A)$ we get  the inhomogeneous wave equation.
Using the limit \eqref{dot-om-decay} and the strong 
 Huygen's principle for the wave  equations, we modify 
the  field part  so that the modified fields satisfy
an inhomogeneous wave equation and  coincide with some soliton fields outside the light cone. Moreover, the modified fields
 are close to the soliton fields for  large time.
The modified  trajectory  satisfies a new system of equations which is a small perturbation of the system (\ref{mls2}) for large times.  

Further, we estimate oscillations of  Hamiltonian and of the Casimir invariant,  
included in the Lyapunov function, along the modified trajectory.
Finally,  we use these estimates and  a lower bound for the Lyapunov function  to prove
that  the angular velocity has the limits
\be\la{ek-lim}
\om(t)\to \om_{\pm}, \qquad t\to\pm\infty.
\ee
The limits imply the global attraction to  particular soliton by \eqref{K-res}.

Similar global attraction is proved in \ci{KS1} for a moving 
 particle without rotation ($\om\equiv 0$) in scalar wave field. 
 In \ci{IKM2004} the result was extended  to the Maxwell--Lorentz system 
with moving particle without rotation.  The global attraction in the case  
when the particles is both moving and rotating is an open problem.

\setcounter{equation}{0}
\section{The Maxwell potentials}
In the Maxwell potentials $A(x,t)=(A_1(x,t),A_2(x,t),A_3(x,t))$ and $\Phi(x,t)$, we have
\be\la{AA}
B(x,t)=\curl A(x,t),\qquad E(x,t)=-\dot A(x,t)-\na \Phi(x,t).
\ee
We choose the Coulomb gauge: $\dv A(x,t)=0$.
Then the  first two lines of the system (\ref{mls1}) are equivalent to the system
\be\la{2mA}
\left\{\ba{rcl}
-\ddot A(x,t)-\na\dot\Phi(x,t)&=&-\De A(x,t)-(\om\we x) \rho(x)
\\
-\De\Phi(t)&=&\rho(x)
\ea\right|.
\ee
Here the second equation   can be solved explicitly:
\be\la{2mA3}
\Phi(x)=-\frac{1}{(2\pi)^3}\int e^{-ik\cdot x}\fr{\hat\rho(k)}{k^2}\,dk=\fr1{4\pi}\int \fr{\rho(y)} {|x-y|} dy.
\ee
 In the projection on divergence-free fields,  the first equation of (\ref{2mA}) is equivalent to the wave equation
\be\la{2mA2}
\ddot A=\De A+\om\we \varrho(x),\qquad \varrho(x):=x\rho(x).
\ee
Then the system  (\ref{mls1}) becomes
\be\la{mls2}
\left\{\ba{rcl}
\dot A(x,t)&=&\Pi(x,t)
\\
\dot \Pi(x,t)&=&\De A(x,t)+\om(t)\we \varrho(x),
\\
 I\dot \om(t)&=&\langle \Pi(x,t)\we\varrho(x)\rangle- \om(t)\we\langle A(x,t)\we\varrho(x)\rangle
 \ea\right|.
\ee
In the last  equation, the terms with $\Phi(x,t)$ cancel:
$\langle  x\we \na\Phi(x,t),\rho(x)\rangle=0$. 
We suppose that  
 $A_0(x)\in C^3(\R^3\setminus {\rm B}_{R_0})$, $\Pi_0(x)\in C^2(\R^3\setminus {\rm B}_{R_0})$, and  
\beqn\nonumber
&&|A_0(x)|+|x|\big( |\na A_0(x)|+|\Pi_0(x)|\big)+ |x|^2\big(|\nabla\nabla A_0(x)|+|\nabla \Pi_0(x)|\big) ={\cal O}(|x|^{-\sigma}),\\ 
\la{A-8} 
 && |\na\na\na A_0(x)|+|\na\na \Pi_0(x)|={\cal O}(|x|^{-1}) ,\quad x\to\infty
 \eeqn 
Evidently, \eqref{A-8} implies \eqref{8}.
\setcounter{equation}{0} 
 \section{Well posedness} 
We first define a suitable phase space. 
We denote the Sobolev spaces $H^s=H^s(\R^3)\otimes\R^3$ with $s\in\R$,
and  $\dot H^1=\dot H^1(\R^3)\otimes\R^3$.
\begin{definition}
{\it i)} $\F=\dot H^1(\R^3)\oplus L^2(\R^3)$ 
is the Hilbert space of   vector  fields $F=(A,\Pi)$ with finite norm
$$
\Vert F\Vert_{\F}:=\Vert \na A\Vert_{L^2(\R^3)}+\Vert \Pi\Vert_{L^2(\R^3)}.
$$
{\it ii)}  $\Y=\F\oplus\R^3$ is the Hilbert space of   Y= $(F, \om)$ with finite norm
$$
\Vert Y\Vert_{\Y}=\Vert F\Vert_{\F}+|\om|
$$
\end{definition}
On $\F$ and $\Y$ we define the local energy seminorms by 
\be\la{6} 
\Vert F\Vert_R=\Vert \na A\Vert_R + \Vert \Pi\Vert_R\quad
{\rm and} \quad
\Vert Y\Vert_R=\Vert F\Vert_R+|\om|
\ee 
for every $R>0$, where $\Vert\cdot\Vert_R$ is the norm in $L^2(B_R)$. 

We write the system (\re{mls2}) as a dynamical equation on $\Y$ 
\be\la{eq-Y} 
\dot Y(t)={\cal F}(Y(t)),\qquad t\in\R, \qquad Y(t) = (A(x,t),\Pi(x,t), \om(t))\in \Y.
\ee 
 \begin{pro}\la{ex} 
For any initial state $Y(0)=(A(x,0),\Pi(x,0),\om(0))\in\Y$, 
the equation {\rm (\ref{eq-Y})}  admits a unique
solution
\be\la{soluY}
Y(t)=(A(x,t),\Pi(x,t),\om(t))\in C(\R,\Y).
\ee
ii) The map $W(t):Y(0)\mapsto Y(t)$ is continuous in $\Y$ for every $t\in\R$.
\smallskip\\
{\rm iii)} The energy is conserved:
\be\la{ecoA}
\cH(Y(t)):=\ds\fr12\int [|\Pi(x,t)|^2+|\curl A(x,t)|^2]dx+\fr12I \om^2(t)=\cH(Y(0)),\quad t\in\R. \ee
iii) The estimate holds, 
\be\la{ov-v} 
|\omega(t)|\le\ov \omega,\quad t\in\R. 
\ee 
 \end{pro}
This proposition is a special version of Proposition 2.2 proved in   \cite{KK2023}.
Denote
\be\la{pi}
\pi:=I\om+\langle\varrho\we A\rangle
\ee 
In \cite{KK2023} was proved that the system \eqref{mls2} admits the functional family of invariants
\be\la{casimir}
C(Y)=f(|\pi|),\qquad  f\in C^1(R)
\ee
 known as {\it Casimir} invariants. 
\setcounter{equation}{0}
 \section{ Kirchhoff representation of solutions} 
Denote by  $\cK_t$ the  $6\times6$ - matrix-valued distributions 
\be\la{mt} 
\cK_t :=\left(\ba{cc} \dot K_t & K_t\\ \ddot K_t&\dot K_t\ea\right),\quad  K_t(x):=\fr 1 {4\pi t}\de(|x|-|t|).
\ee 
By the Kirchhoff formula, 
$$
F(x,t)=F_{r}(x,t)+F_{K}(x,t),\quad  F_{r}(x,t)=\begin{pmatrix} A_r(x,t),\Pi_r(x,t)\end{pmatrix},\quad   F_{K}(x,t)=\begin{pmatrix} A_K(x,t),\Pi_K(x,t)\end{pmatrix},
$$ 
where
\beqn\la{A20} 
\left(\ba {c}A_{r}(x,t)\\ \Pi_{r}(x,t)\ea\right) 
=\int_0^t \Big[\cK_{t-s}(x) *\left(\ba{c} 0\\ \om(s)\we\varrho(x)\ea\right)\,ds; 
\eeqn 
\beqn\la{A19} 
\left(\ba {c}A_{K}(x,t)\\\Pi_{K}(x,t)\ea\right) 
=\cK_t(x) 
*\left(\ba{c} A_0(x)\\ \Pi_0(x)\ea\right). 
\eeqn 
\begin{lemma}\la{WPI00} (cf. [Lemma 3.3]\cite{KS1})
Let $(A_0,\Pi_0)$  satisfies (\re {A-8}) with some $\si >1/2$. Then 
\be\la{Fk-decay}
 |\na A_K(x,t)|+|\Pi_K(x,t)|=\cO( t^{-1-\si}),\quad t\to\infty.
\ee
\end{lemma}
\begin{proof}
By \eqref{A19},
\beqn\nonumber
&&A_K(x,t)=\fr 1{4\pi t}\int\limits_{S_t(x)}\Pi_0(y)~d^2y+
 \fr \pa {\pa t}~\Big(~\fr 1{4\pi t}\int\limits_{S_t(x)}A_0(y)~d^2y~~\Big)\\
\nonumber
&&=\fr{t}{4\pi}\int\limits_{S_1(0)}~\Pi_0(x+tz)\,d^2z+\fr{1}{4\pi}\int\limits_{S_1(0)}A_0(x+tz)\,d^2z
+\fr{t}{4\pi}\int\limits_{S_1(0)}\na A_0(x+tz)\cdot z\,d^2z.
\eeqn
\beqn\nonumber
\na A_K(x,t)&=&\fr{t}{4\pi}\int_{S_1(0)}\na_x\Pi_0(x+tz)\,d^2z+\fr{1}{4\pi}\int_{S_1(0)}\na_x A_0(x+tz)\,d^2z\\
\nonumber
&+&\fr{t}{4\pi}\int_{S_1(0)}\na_x(\na A_0(x+tz)\cdot z)\,d^2z.
\eeqn
Here $S_t(x)$ denotes the sphere $\{y:~|y-x|=t\}$.
Here all derivatives are understood in the classical sense. A similar representation holds for $\Pi_K(x,t)=\dot A_K(x,t)$.
Hence, taking into account the assumption (\re{A-8}), we obtain  for $t>R_0+|x|$,
\beqn\la{WPbd}
|\na A_K(x,t)|+|\Pi_k(x,t)||&\le& C \sum_{s=0}^1 t^{s}\int_{S_1(0)}|x+tz|^{-\sigma-1-s}\,d^2z.
\eeqn
For $\al\not=2$  one has 
\be\la{alpha}
\int_{S_1(0)}|x+tz|^{-\al}\,d^2z=\fr{2\pi}{(\al-2)|x|t}\Big((t-|x|)^{-\al+2}-(t+|x|)^{-\al+2}\Big)=\cO(t^{-\al}) 
\ee
as $t\to\infty$ since for any $\beta\in\R$,
$$
(t+|x|)^{\beta}-(t-|x|)^{\beta}=t^{\beta}\big[(1+\frac{|x|}t)^{\beta}-(1-\frac{|x|}t)^{\beta}\big]\sim  2\al |x|t^{\beta-1}
$$
Finally, \eqref{Fk-decay} follows by \eqref{WPbd} and \eqref{alpha}.
\end{proof}
\setcounter{equation}{0}
\section{Orbital stability of solitons}
Solitons of the system (\ref{mls2}) are stationary solutions 
\be\la{solvom}
Y_{\omega}(x)=S_{\omega}(x)=(A_{\om}(x), 0,\om).
\ee
Substituting  into \eqref{mls2}, we get
\be\la{solit2}
\Delta  A_{\omega}(x)=-\omega\we\varrho(x).
\ee
In the Fourier representation
\be\la{solit3}
\hat A_{\omega}(k)=\fr{\omega\we\hat\varrho(k)}{k^2}.
\ee
By \eqref{rosym}, $A_{\omega}\in H^s$ for any $\om\in\R^3$ and any $s\in\R$.

One has
\be\la{kappa}
\langle\varrho\we A_{w}\rangle=\langle\hat\varrho\we\hat A_{\om}\rangle=\langle\hat\varrho\we(\frac{\om \we \hat \varrho}{k^2})\rangle
=\om\langle\frac{\hat\varrho}{k^2},\hat\varrho\rangle-\langle\frac{\hat\varrho}{k^2},\om\cdot\hat\varrho\rangle
=\frac 23\varkappa_0\om,
\ee
where we denote
\be\la{kappa0}
\varkappa_{0}=\int\frac{|\hat\varrho(k)|^2}{k^2} dk=\sum\limits_{j=1}^3\int\frac{|\hat\varrho_j(k)|^2}{k^2} dk. 
\ee
Hence, 
\be\la{M0}
\pi_\om=I\om+\langle\hat\varrho\we \hat A_{\om}\rangle=\om(I+\frac 23\varkappa_0)=\om I_{\eff},\quad 
 I_{\eff}:=I+\frac 23\varkappa_0
\ee
 Denote 
\be\la{Lam}
\Lambda_\om={\cal H}(Y)-|\om||\pi|.
\ee
In \ci{KK2023}  the following  result on orbital stability is proved.
\begin{pro}\la{A} \ci[Proposition 4.1]{KK2023} 
Let condition \eqref{rosym} holds. 
Then   the following upper bound holds with an $\al> 0$: 
\be\la{o-s}
\delta\Lambda_{\om}:=\Lambda_{\om}(S_\om + \de Y)-\Lam_{\om}(S_\om)\ge \al\Vert \de Y\Vert_{\Y}^2,\quad
\de Y\in\Y,\quad \Vert Y\Vert_{\Y}\ll 1.
\ee 
 \end{pro} 
 \setcounter{equation}{0} 
 \section{Main result} 
 The main results of the paper is the following  theorem. 
 \begin{theorem}\la{A1} 
 Let conditions \eqref{rosym}   and \eqref{mm} hold,  $m_b\not\in {\cal M}_{\rho}$. 
Let $Y(t)=(A(t),\Pi(t),\om(t))\in C(\R,\Y)$ be the solution to \eqref{mls2} with initial data
$Y_0=(A_0\,, \Pi_0,\om_0)\in \Y$  satisfying \eqref{A-8}
with some $\sigma>1/2$. Then  for every $R>0$, 
  \be\la{10} 
\lim\limits_{t\to\pm\infty} \Vert Y(t)-Y_{\pm\om}\Vert_R=0.
\ee 
\end{theorem} 
By \eqref{K-res}, for the  prove \eqref{10} it suffices to prove the existence of the limits (\ref{ek-lim}).  
We combine the bound \eqref{o-s} and the relaxation of the acceleration \eqref{dot-om-decay}
with a Hamiltonian structure for the  system (\ref{mls2}). 
We prove (\ref{ek-lim})  only for $t\to+\infty$ since the system is time-reversal. Introduce  
$$ 
{\rm osc}_{[T;+\infty)}\om(t):=\sup\limits_{t_1,t_2\ge T}|\om(t_1)-\om(t_2)|. 
$$  
The existence of the limits (\ref{ek-lim}) follows from the following proposition.  
\begin{pro}\la{5.1} 
Let the assumptions of Theorem \re{A1} be fulfilled. 
Then  
\be\la{osc} 
{\rm osc}_{[T;+\infty)}\om(t)\to 0\,\,\,{\rm as}\,\,\,T\to{+\infty}. 
\ee 
\end{pro} 
 The idea of the proof is as follows. We modify the trajectory  of the solution. 
 The new trajectory  satisfy a new system of equations which is a small perturbation of the system (\ref{mls2}) for large $t$.  
\setcounter{equation}{0} 
\section{Modification of the trajectory}
By \eqref{dot-om-decay}, for every $\ve>0$ there exists $t_\ve$ such that 
\be\la{5.7} 
|\dot \om(t)|\le \ve \quad {\rm for}~~~t\geq t_\ve, \quad
{\rm and}~~t_\ve\to\infty~~{\rm as}~~\ve\to 0\,. 
\ee 
Let us consider the point 
\be\la{moms} 
t_{1,\ve}=t_{\ve}+1,\quad  t_{2,\ve}=t_{1,\ve}+R_\rho, \quad t_{3,\ve}=t_{2,\ve}+R_\rho. 
\ee 
We set 
$$ 
\rw_{\ve}=\om(t_{\ve}).
$$ 
From  (\re{5.7}) it follows  that there exists 
$\om_\ve(t)\in C^1(\R)$ such that 
\beqn\la{5.8} 
\om_\ve(t)= 
\left\{ 
\ba{lll} 
\om(t)&{\rm \,\,for\,\,}&t\in[t_{1,\ve},+\infty)\,,\\ 
\rw_{\ve}& {\rm \,\,for\,\,}&t\in(-\infty,t_{\ve}]\,, 
\ea 
\right. 
\eeqn 
 and 
\be\la{5.9} 
|\dot \om_\ve(t)|\leq C\ve{\rm \,\,\,for\,\,all\,\,}t\in\R 
\ee 
with $C>0$ independent of $\ve\in (0,1)$. 
Now we set 
\beqn\la{A20eps} 
F_\ve(x,t)=\begin{pmatrix}A_{\ve}(x,t)\\ \Pi_{\ve}(x,t)\end{pmatrix}
=\int_{-\infty}^t
\cK_{t-s}(x) *\left(\ba{c} 0\\ \om_\ve(s)\we\varrho(x)\ea\right)\, ds,\quad t>0. 
\eeqn 
\subsection{Modified fields}
Here we show that the modified fields satisfy the inhomogeneous wave equation, coincide with soliton fields outside a certain light cone, 
and coincide with the retarded fields $F_{r}$ defined in \eqref{A20} in a smaller light cone.  
\begin{lemma}\la{44.2} 
i) The fields $F_\ve$  coincide with a soliton 
outside a light cone: 
\be\la{ret-sol} 
F_\ve(x,t)=F_{\rw_\ve}(x)=\begin{pmatrix} A_{\rw_\ve}(x)\\0\end{pmatrix},\quad {\rm for}~~|x|>t-t_\ve+R_{\rho}. 
\ee 
ii) $F_\ve(x,t)$ satisfies the system 
\be\la{Bve}
\dot F_{\ve}(x,t)=\begin{pmatrix} 0 & 1 \\ \Delta & 0\end{pmatrix} F_\ve(x,t)+\begin{pmatrix} 0\\ \om_\ve(t)\we\varrho(x)\end{pmatrix}
\quad  t\in\R, \quad x\in\R^3.
\ee
iii) The field $F_{\ve}(x,t)$ coincide with $F_{r}(x,t)$ in 
the light cone $\{|x|<t-t_{2,\ve}\}$. 
\end{lemma} 
\begin{proof} 
i) Consider the soliton fields $F_{\rw_\ve}(x)$ as the solution of the Cauchy problem for system 
(\re{Bve}) with initial data $F^{-T}=F_{\rw_\ve}(x))$ at $-T$, where $T>0$.

Applying  the formulas of type (\re{A20})--(\re{A19}), we obtain 
\be\la{solT} 
F_{\rw_\ve}(x)=\int_{-T}^t \cK_{t-s}(x) *\left(\ba{c} 0\\ \rw_\ve\we\varrho(x)\ea\right)\,ds +\cK_{t+T} *F_{\rw_\ve}.
\ee 
By  (\re{Fk-decay}),  the last summand in \eqref{solT} tends to zero in 
$H^1_{\rm loc}(\R^3)\oplus L^2_{\rm loc}(\R^3)$   as $T\to+\infty$.
 Hence, proceeding to the 
limit as $T\to+\infty$ we obtain the identity of distributions, 
\be\la{solinf} 
F_{\rw_\ve}(x)=\int_{-\infty}^t \cK_{t-s}(x) *\left(\ba{c} 0\\ \rw_\ve\we\varrho(x)\ea\right)\,ds. 
\ee 
By \eqref{rosym} and \eqref{mt}, the integration domain   in the inner integrals of \eqref{A20eps}  and \eqref{solinf} is 
\be\la{dom}
D(s)= \{y\in B_{R_{\rho}} : t-s=|x-y|\}.
\ee
In the region $|x|>t-t_\ve+R_{\rho}$, one has $|x-y|> t-t_\ve+R_{\rho}-R_{\rho}= t-t_\ve$.
Comparing with \eqref{dom}, we  get $s<t_\ve$. Hence, 
in the region $|x|>t-t_\ve+R_{\rho}$, the right hand side of (\re{solinf})
 coincides with  (\re{A20eps}) by (\re{5.8}). 

ii)  By i) it suffices to prove that $F_\ve$ is a solution to \eqref {Bve} in $(x,t)\in K_0=\{|x|<t-t_{\ve}+R_{\rho}\}$. 
The strong Huygen's principle  implies 
\beqn\la{A22eps} 
F_{\ve}(x,t)=\int_{t_\ve-2R_{\rho}}^t\cK_{t-s}(x) *\left(\ba{c} 0\\ \om_\ve(s)\we\varrho(x)\ea\right)\,ds,\quad  (x,t)\in K_0 
\eeqn 
since in this case  $|x-y|<t-t_{\ve}+2R_{\rho}$ for $|y|<R_{\rho}$, and $D(s)$ is non-empty only for  $s>t_{\ve}-2R_{\rho}$.
 
For $T>|t_{\ve}-2R_{\rho}|$  introduce the field 
\beqn\la{A23eps} 
\ti F_T(x,t)=\int_{-T}^t\cK_{t-s}(x) *\left(\ba{c} 0\\ \om_\ve(s)\we\varrho(x)\ea\right)\,ds+\cK_{t+T}*F_{\rw_\ve}
\eeqn 
The field $\ti F_T(x,t)$ satisfies  (\re{Bve}) for $x\in\R^3$ and $t\in\R$ by the same argument as in the proof of i). 
Finally, the second summand in the right hand side of (\re{A23eps}) tends to zero as $T\to+\infty$ like in (\re{solT}). 
Hence, $F_{\ve}$ satisfy (\re{Bve})  in $K_0$. 
 
iii) follows from (\re{mt})--(\re{A20}),  (\re{5.8}) and (\re{A20eps}). Namely, in this case for $|y|<R_{\rho}$ we get 
$|x-y|<t-t_{2,\ve}+R_{\rho}=t-t_{1,\ve}$. Hence, $D(s)$ is non-empty only for $s>t_{1,\ve}$, where $\om_\ve(t)=\om(t)$ by \eqref{5.8}.
\end{proof}
\smallskip
Note that $F_{\ve}(x,t_{3,\ve})=F_{\rw_\ve}(x)$  outside the ball $B_{3R_{\rho}+1}$ by Lemma \ref{44.2}-i). 
Now we show that at $t=t_{3,\ve}$  the modified fields $F_{\ve}(x,t_{3,\ve})$ are sufficiently close to  the soliton $F_{\rw_\ve}(x)$ 
in the ball $B_{3R_{\rho}+1}$.
\begin{lemma}
\be\la{fin} 
\Vert F_{\ve}(\cdot,t_{3,\ve})-F_{\rw_\ve}(\cdot)\Vert_{H^1(B_{3R_{\rho}+1})\oplus L^2(B_{3R_{\rho}+1})}={\cal O}(\ve). 
\ee 
\end{lemma}
\begin{proof}
By \eqref{A20eps} and \eqref{solinf},
$$ 
F_{\ve}(x,t_{3,\ve})-F_{\rw_{\ve}}(x)\
= 
\int\limits^{t_{3,\ve}}_{t_{\ve}}\cK_{t_{3,\ve}-s}(x)*\left(\ba{c} 0\\ (\om_\ve(s)-\rw_\ve)\we\varrho(x)\ea\right)\, ds. 
$$
Hence,  (\re{5.8}), (\re{5.9}) imply \eqref{fin}.  
\end{proof}
\begin{cor}
\be\la{fin+} 
\Vert F_{\ve}(\cdot,t_{3,\ve})-F_{\rw_\ve}(\cdot)\Vert_{H^1(\R^3)\oplus L^2(R^3)}={\cal O}(\ve). 
\ee 
\end{cor}
\subsection{Asymptotics for large $t$}
First,  we express the Lorentz force equation   for $t\ge T_\ve:=t_{3,\ve}$ in terms of the field $F_{\ve}$. In this region $\om_{\ve}(t)=\om(t)$. 
Thus, we can change $\om_{\ve}(t)$ by $\om(t)$ in the equation (\re{Bve})  for $F_{\ve}$: 
$$
\dot F_{\ve}(x,t)=\begin{pmatrix} 0 & 1 \\ \Delta & 0\end{pmatrix} F_\ve(x,t)+\begin{pmatrix} 0\\ \om(t)\we\varrho(x)\end{pmatrix},\quad t\ge T_\ve.
$$
Further, one has $F_{\ve}(x,t)=F_{r}(x,t)$  inside   the light cone $\{|x|<t-t_{2,\ve}\}$ by Lemma \re{44.2}, iii). 
Thus, for $t\ge T_\ve$ in ${\rm supp}\,\rho(x)$ we have $F=F_{\ve}+F_{K}$,  and hence 
\be\la{2eps} 
 I\dot \om(t)=\langle\Pi_{\ve}(x,t)\we \varrho(x)\rangle+\om(t)\we\langle \varrho(x)\we A_{\ve}(x,t)\rangle+f(t),\,\,\,t\ge T_\ve, 
\ee
where 
$$ 
f(t):=\langle \Pi_{K}(x,t)\we \varrho(x)\rangle+\om(t)\we\langle  \varrho(x)\we A_{K}(x,t)\rangle. 
$$
\begin{lemma} Let $(A_0,\Pi_0)$  satisfies (\re {A-8}) with some $\si >1/2$. Then
\be\la{asf} 
f(t)={\cal O}(t^{-1-\si}),\quad t\to \infty.   
\ee
\end{lemma} 
\begin{proof} The asymptotics follow from the condition \eqref{rosym},  the bound (\re{ov-v}) and the  asymptotics  (\re{Fk-decay}). 
In particular, \eqref{rosym} and  (\re{Fk-decay}) imply
$$
|\langle  \varrho(x)\we A_{K}(x,t)\rangle| \le \Vert\varrho(x)\Vert_{R_\rho}\Vert A_{K}(x,t)\Vert_{R_\rho}
\le C(\rho)\Vert \na A_{K}(x,t)\Vert_{R_\rho}\le  C(\rho)t^{-1-\si}
$$
by the  Sobolev embedding $\dot H^1(B_{R_\rho})\subset L^2(B_{R_\rho})$. 
\end{proof}
Now we  prove that for large $t$,
${\cal H}(t)$ and $|\pi(t)|$ are ``almost conserved'' along a trajectory 
$Y_{\ve}(t)=(A_\ve(t),\Pi_\ve(t),\om(t))$. 
\begin{lemma} Let $(A_0,\Pi_0)$  satisfies (\re {A-8}) with some $\si >1/2$.
Then  the oscillations of the Hamiltonian and the total momentum are small for large $T$ and  $t>T$,
\beqn\la{astiH} 
{\cal H}(Y_{\ve}(t))&=&{\cal H}(Y_{\ve}(T))+{\cal O}(T^{-\si}),\\ 
\la{asP} 
|\pi(Y_{\ve}(t))|&=&|\pi(Y_{\ve}(T))|+{\cal O}(T^{-\si}). 
\eeqn 
\end{lemma} 
\begin{proof} 
\eqref{mls2},  \eqref{2eps}, \eqref{asf}  implies for $t> T$
\beqn\nonumber 
\fr{d}{dt}{\cal H}(Y_{\ve}(t))&=&\frac 12\fr{d}{dt}\left(I\om^2(t)+\int(\Pi_{\ve}^2(x,t)+{\rm curl} A_\ve^2(x,t))dx\right)\\
\nonumber
&=&I\om(t)\cdot\dot \om(t)+\langle \Pi_\ve(x,t),\dot \Pi_\ve(x,t)\rangle+\langle {\rm curl} A_\ve(x,t),{\rm curl} \dot A_\ve(x,t)\rangle\\ 
\nonumber
&=&\om(t)\cdot\big(\langle \Pi_{\ve}(x,t)\we \varrho(x)\rangle+\om(t)\we\langle \varrho(x)\we A_{\ve}(x,t)\rangle+f(t)\big)\\
\nonumber
&+&\langle \Pi_{\ve}(x,t),\Delta A_{\ve}(x,t)+\om\we\varrho(x)\rangle+\langle {\rm curl}\, A_\ve(x,t), {\rm curl}\, \Pi_\ve(x,t)\rangle\\
\label{H-dot}
&=&\om(t)\cdot f(t)={\cal O}(t^{-1-\si}).
\eeqn 
since
$$
\om\cdot\big(\om\we\langle \varrho\we A_{\ve}\rangle\big)=0,~
\langle \Pi_{\ve},\Delta A_{\ve}\rangle+\langle {\rm curl}\, A_\ve, {\rm curl}\, \Pi_\ve\rangle=0,~
\om\cdot\langle \Pi_{\ve}\we \varrho\rangle\big)+\langle \Pi_{\ve},\om\we\varrho\rangle=0
$$
Similarly, for $t> T$, one has  
\beqn\nonumber
&&\pi(Y_{\ve})\cdot \fr{d}{dt}\pi(Y_{\ve})=(I\om+\langle\varrho\we A_{\ve}\rangle)\cdot\fr{d}{dt}(I\om+\langle\varrho\we A_{\ve}\rangle)\\
\nonumber
&&=(I\om+\langle\varrho\we A_{\ve}\rangle)\cdot\big(\langle \Pi_{\ve}\we \varrho\rangle+\om(t)\we\langle \varrho\we A_{\ve}\rangle+f
+\langle\varrho\we \Pi_{\ve}\rangle\big)\\
\nonumber
&&=(I\om+\langle\varrho\we A_{\ve}\rangle)\cdot\big(\om(t)\we\langle \varrho\we A_{\ve}\rangle+f\big)
=(I\om+\langle\varrho\we A_{\ve}\rangle)\cdot f=\pi(Y_{\ve})\cdot f
\eeqn
Hence,
\be\la{pi-dot}
\fr{d}{dt}|\pi(Y_{\ve}(t))|=\frac{\pi(Y_{\ve}(t))\cdot \fr{d}{dt}\pi(Y_{\ve}(t))}{|\pi(Y_{\ve}(t))|}=\frac{\pi(Y_{\ve}(t))\cdot f(t)}{|\pi(Y_{\ve}(t))|}={\cal O}(t^{-1-\si}).
\ee
Finally, \eqref{H-dot} and \eqref{pi-dot} imply (\re{astiH}) and  (\re{asP}), respectively. 
\end{proof} 
\setcounter{equation}{0} 
\section{Proof of Proposition \re{5.1}}
  For $t\ge T_{\ve}:=t_{3,\ve}$, one has 
 \be\la{piY}
  \pi(Y_{\ve}(t))=\pi(Y_{\ti \om(t)})=\ti\om(t)I_{\rm eff},~~{\rm where} ~~\ti\om(t):=\frac{\pi(Y_{\ve}(t))}{I_{\rm eff}}.
  \ee
From (\re{asP})  it follows that 
\be\la{osctiv} 
{\rm osc}_{[T,+\infty)}|\ti \om(t)|\to0\,\,\,{\rm as}\,\,\,T\to+\infty. 
\ee
Definition \eqref{pi} together with \eqref{5.8}--\eqref{5.9} and  \eqref{fin+} imply
$$
  \pi(Y_{\ve}(T_\ve))=I\om(T_{\ve})+\langle \varrho, A_{\ve}(T_{\ve})\rangle=
 I\rw_{\ve} +\langle \varrho, A_{\rw_{\ve}}\rangle +{\cal O}(\ve)=\rw_\ve I_{\rm eff}+{\cal O}(\ve)
 $$
Comparing with \eqref{piY}, we get  
\be\la{ww}
\ti \om(T{_\ve})-\rw_{\ve}={\cal O}(\ve).
\ee 
Together with (\re{osctiv}) this implies the bound $|\ti \om(t)|\le \ov \om_1$ for $t\ge T$. 

Now we apply the orbital stability estimate (\re{o-s}) and get  
\be\la{lo-bound-t} 
\alpha(\Vert F_{\ve}(\cdot,t)-F_{\ti \om(t)}\Vert^2_{\cal F})\le \Lambda_{\ti \om(t)}(Y_{\ve}(t))-\Lambda_{\ti \om(t)}(Y_{\ti \om(t)})
={\cal H}(Y_{\ve}(t))-{\cal H}(Y_{\ti \om(t)}).
\ee 
by  \eqref{Lam} and \eqref{piY}.
\begin{lemma}\la{rhs} 
The right hand side of (\re{lo-bound-t}) is arbitrary small uniformly in $t\ge T$ 
for sufficiently small $\ve$ and sufficiently large $T$.
\end{lemma} 
\begin{proof}
One has
$$
{\cal H}(Y_{\ve}(t))-{\cal H}(Y_{\ti \om(t)})=[{\cal H}(Y_{\ve}(t))-{\cal H}(Y_{\ve}(T))]
-[{\cal H}(Y_{\ti\om(t)})-{\cal H}(Y_{\ti\om(T)})]+[{\cal H}(Y_{\ve}(T))-{\cal H}(Y_{\ti \om(T)})],
$$
where 
${\cal H}(Y_{\ve}(t))-{\cal H}(Y_{\ve}(T))={\cal O}(T^{-\si})$ by (\re{astiH}),
and ${\cal H}(Y_{\ti\om(t)})-{\cal H}(Y_{\ti\om(T)})={\cal O}(\ve)$ by \eqref{osctiv} and \eqref{ww}.
Finally, ${\cal H}(Y_{\ve}(T))-{\cal H}(Y_{\ti \om(T)})$ is small
by (\re{ret-sol}), (\re{fin+}) and \eqref{ww}.
\end{proof}
 From this lemma it follows that  
$$
{\rm osc}_{[T,+\infty)}\Vert F_{\ve}(\cdot,t)\Vert_{\cal F} \to 0,\quad   T\to+\infty.
$$ 
 Indeed, 
 $$
 F_{\ve}(t_2)-F_{\ve}(t_1)= 
(F_{\ve}(t_2)-F_{\ti \om(t_2)})-(F_{\ve}(t_1)-F_{\ti \om(t_1)})+ (F_{\ti \om(t_2)}-F_{\ti \om(t_1)}). 
$$ 
For $t_1,t_2>T$ the first and the second summands are small by (\re{lo-bound-t}) and the lemma, the third summand is small by (\re{osctiv}), 
since the soliton field $F_\om$ depends continuously on $\om$ in $\cal F$.
 
Together with (\re{asP}) this implies ${\rm osc}_{[T,+\infty)}\om(t)\to0$ as $T\to+\infty$ and hence (\re{osc}) follows. 
Proposition \re{5.1} is proved.\hfill $\Box$ 
 

{\large Institute for Information Transmission Problems, RAS }\\
{\it E-mail address:} ek@iitp.ru
\smallskip

{\large Institute for Information Transmission Problems, RAS }\\
{\it E-mail address:} akomech@iitp.ru

\begin{thebibliography}{99} 

\bibitem{IKM2004}
V. Imaykin, A. Komech, N. Mauser, 
Soliton-type asymptotics for the coupled Maxwell--Lorentz equations, 
{\em Ann. Inst. Poincar\'e, Phys. Theor.} {\bf  5} (2004), 1117--1135.

\bibitem{IKS} 
V.M. Imaikin, A.I. Komech, H. Spohn, 
Soliton-type asymptotics and scattering for a charge coupled  to the Maxwell field, 
{\em Russian J. Math. Phys.} {\bf 9} (2002), no.4, 428-436. 
 
\bibitem{IKS2004}
V. Imaykin, A. Komech, H. Spohn, 
Rotating charge coupled to the Maxwell field: scattering theory and adiabatic limit, 
{\em Monatsh. Math.} {\bf 142} (2004), no. 1--2,  143--156.

\bibitem{KS1998}
A. Komech, H. Spohn, 
Soliton-like asymptotics for a scalar particle interacting with a scalar wave field, 
{\em Nonlinear Analysis} {\bf 33} (1998), no. 1, 13--24. 

\bibitem{KS2000}
A.I. Komech, H. Spohn, 
Long-time asymptotics for the coupled Maxwell--Lorentz equations, 
{\em Comm. Partial Differential Equations} {\bf 25} (2000), 559--584.

\bibitem{KSK} 
A.I. Komech, H. Spohn, M. Kunze, 
Long-time asymptotics for a classical particle interacting with a scalar wave field, 
{\em Comm. Partial Differential Equations} {\bf 22} (1997), no.1/2, 307-335. 

\bibitem{KK2023}
A. Komech, E. Kopylova,
On the stability of solitons for Maxwell--Lorentz equations with rotating particle,
{\em Milan Journal of Math.}  {\bf 91} (2023), 155--173.


\bibitem{KS1} 
A.I. Komech, H. Spohn, 
Soliton-like asymptotics for a classical particle interacting with a scalar wave field, 
{\em Nonlinear Analysis, Theory, Methods $\&$ Applications} {\bf 33} (1998), no.1, 13-24. 
 
\bibitem{Kunze} 
M. Kunze, 
On the absence of radiationless motion for a rotating classical charge 
{\em Advances in Mathematics}  {\bf 223}, (2010), no. 5, 1632--1665.

\bibitem{S2004}
 H. Spohn, 
 Dynamics of Charged Particles and Their Radiation Field, Cambridge
University Press, Cambridge, 2004.


 \end{thebibliography}
\end{document}